\documentclass[11pt, reqno]{article}

\pdfoutput=1

\usepackage{graphicx}
\usepackage{jheppub}
\usepackage{amssymb}
\usepackage{amsmath}
\usepackage[usenames,dvipsnames]{xcolor}
\usepackage{epsfig}
\usepackage{dcolumn}
\usepackage{tikz}
\usetikzlibrary{shapes.geometric, arrows}
\usepackage{upgreek}
\usepackage{setspace}
\usepackage{enumitem}
\usepackage{array,multirow,bigdelim,arydshln}
\usepackage{appendix}
\usepackage{xparse}
\usepackage[utf8]{inputenc}
\usepackage{hyperref}
\hypersetup{
	colorlinks,
	urlcolor=Maroon,
	linkcolor=Maroon,
	citecolor=Maroon
}

\usepackage{amsthm}
\usepackage{amsmath}
\newtheorem{thm}{Theorem}[section]
\newtheorem{prop}[thm]{Proposition}

\newtheorem{cor}[thm]{Corollary}
\newtheorem{defn}[thm]{Definition}
\theoremstyle{definition}
\newtheorem{example}[thm]{Example}

\newtheorem{conjecture}[thm]{Conjecture}
\theoremstyle{remark}

\usepackage{mathtools}
\usepackage{scalerel}

\usepackage{float}
\restylefloat{table}

\NewDocumentCommand{\binomial}{omm}
{%
	\genfrac(){0pt}{}{#2}{#3}%
	\IfValueT{#1}{_{\!#1}}%
}
\NewDocumentCommand{\eulerian}{omm}
{%
	\genfrac<>{0pt}{}{#2}{#3}%
	\IfValueT{#1}{_{\!#1}}%
}

\def \s {\sigma}

\usepackage{underscore}

\usepackage{latexsym}
\usepackage{tikz}

\usepackage[percent]{overpic}

\usepackage{multirow} 
\usepackage{slashed}

\usepackage[]{shuffle}

\def\yz#1\yz {{\color{blue} [[YZ: #1]] }}

\def\yzg#1\yzg {{\color{gray} [[YZ: #1]] }}

\def\yzz#1\yzz {{\color{gray} [[To be verified: #1]] }}

\def\ne#1\ne {{\color{green} [[NE: #1]] }}

\usepackage{longtable}

\usepackage{cleveref}

\usepackage{diagbox}

\usepackage{blkarray}

\title{The CEGM NLSM}

\author[a]{Nick Early,}

\affiliation[a]{Institute for Advanced Study, Princeton, USA.}

\emailAdd{earlnick@ias.edu}

\abstract{
		Studying quantum field theories through geometric principles has revealed deep connections between physics and mathematics, including the discovery by Cachazo, Early, Guevara and Mizera (CEGM) of a generalization of biadjoint scalar amplitudes.  However, extending this to generalizations of other quantum field theories remains a central challenge.  Recently it has been discovered that the nonlinear sigma model (NLSM) emerges after a certain zero-preserving deformation from $\text{tr}(\phi^3)$.  In this work, we find a much richer story of zero-preserving deformations in the CEGM context, yielding generalized NLSM amplitudes.  We prove an explicit formula for the residual embedding of an $n$-point NLSM amplitude in a mixed $n+2$ point generalized NLSM amplitude, which provides a strong consistency check on our generalization.  We show that the dimension of the space of pure kinematic deformations is $\gcd(k,n)-1$, we introduce a deformation-compatible modification of the Global Schwinger Parameterization, and we include a new proof, using methods from matroidal blade arrangements, of the linear independence for the set of planar kinematic invariants for CEGM amplitudes.  Our framework is compatible with string theory through recent generalizations of the Koba-Nielsen string integral to any positive configuration space $X^+(k,n)$, where the usual Koba-Nielsen string integral corresponds to $X(2,n) = \mathcal{M}_{0,n}$. 
}

\begin{document}
	\maketitle
	\addtocontents{toc}{\protect\setcounter{tocdepth}{1}}
	\def \tr {\nonumber\\}
	\def \nn {\nonumber}
	\def \la {|}
	\def \ra {|}
	\def \lan {\langle}
	\def \ran {\rangle}
	\def \dd {\Theta}
	\def\hset{\texttt{h}}
	\def\gset{\texttt{g}}
	\def\sset{\texttt{s}}
	\def \be {\begin{equation}}
		\def \ee {\end{equation}}
	\def \ba {\begin{eqnarray}}
		\def \ea {\end{eqnarray}}
	\def \bg {\begin{gather}}
		\def \eeg {\end{gather}}
	\def \k {\kappa}
	\def \h {\hbar}
	\def \r {\rho}
	\def \l {\lambda}
	\def \be {\begin{equation}}
		\def \en {\end{equation}}
	\def \bes {\begin{eqnarray}}
		\def \ens {\end{eqnarray}}
	\def \red {\color{Maroon}}
	\def \pt {{\rm PT}}
	\def \s {\mathfrak{s}}
	\def \t {\mathfrak{t}}
	\def \v {\mathfrak{v}}
	\def \C {\textsf{C}}
	\def \tp {||}
	\def \p {x}
	\def \x {z}
	\def \V {\textsf{V}}
	\def \ls {{\rm LS}}
	\def \ma {\Upsilon}
	\def \SL {{\rm SL}}
	\def \GL {{\rm GL}}
	\def \w {\omega}
	\def \e {\epsilon}
	
	\numberwithin{equation}{section}

	
\section{Introduction}

In recent years, geometric approaches to scattering amplitudes have not only deepened our conceptual understanding of quantum field theory but also unveiled surprising direct connections between ``toy model'' theories such as the biadjoint scalar or $\text{tr}(\phi^3)$ amplitude, and the realistic physics of pion scattering by way of a kinematical shift.  Similar patterns appear in broader contexts: just as CHY amplitudes can connect $\text{tr}(\phi^3)$ to the NLSM through deformations, it is natural to find a generalization of this connection in the context of a generalization of quantum field theory that we have been uncovering for the past several years, in work of Cachazo-Early-Guevara and Mizera (CEGM)   \cite{Cachazo:2019ngv} on generalized biadjoint scalar amplitudes, and \cite{Arkani-Hamed:2019mrd} on generalized string integrals, now finally making contact with real-world physics.  The main construction of \cite{Arkani-Hamed:2024nhp}, foreshadowed in \cite{Arkani-Hamed:2023lbd}, implemented a kinematic deformation of the $\text{tr}(\phi^3)$ amplitude to discover to leading order in $\delta \rightarrow \infty$ the NLSM amplitude.  We start from the construction of the Koba-Nielsen string integral and build up to CEGM amplitudes, before introducing deformations.

The Koba-Nielsen string \cite{Koba:1969kh} integral is an integral over the positive Grassmannian $G^+(2,n)$ modulo column rescaling
\begin{eqnarray}\label{eq: Koba-Nielsen}
	I^{(2)}_n(s) & = & \alpha'^{(n-3)}\int_{\mathcal{M}^+_{0,n}}\frac{d^{n-3}z}{z_{1,2}z_{2,3}\cdots z_{n,1}}\prod_{i<j}(z_{i,j})^{\alpha' s_{ij}},
\end{eqnarray}
where $z_{i,j} = z_i-z_j$, and the biadjoint scalar partial amplitude \cite{Cachazo:2013hca}, or the $\text{tr}(\phi^3)$ amplitude, is the limit $m^{(2)}_n(s) =\lim_{\alpha' \rightarrow 0} I^{(2)}_n(s).$ 
More generally, for any $(k,n)$, the stringy integral \cite{Arkani-Hamed:2019mrd} is given by 
\begin{eqnarray}
	I^{(k)}_n(s) & = &  (\alpha')^{(k-1)(n-k-1)}\int_{G_+(k,n)\slash T}\omega_{k,n}R_{k,n},
\end{eqnarray}
where 
$$\omega_{k,n} = \Omega(G_+(k,n)\slash T) = \frac{d^{k\times n}}{\text{vol}SL(k)\times GL(1)^n}\frac{1}{p_{12\cdots k}\cdots p_{n12\cdots k-1}}$$
is the differential volume form on the positive configuration space, and $R_{k,n} = \prod_{j_1<\cdots <j_k}p_{j_1\cdots j_k}^{\alpha'\mathfrak{s}_{j_1\cdots j_k}}$
is the regulator and $p_{j_1\cdots j_k}$ are the maximal $k\times k$ Pl\"{u}cker coordinates on the Grassmannian $G(k,n)$.

The standard generalized biadjoint scalar partial amplitude $m^{(k)}(\mathbb{I}_n,\mathbb{I}_n)$, where $\mathbb{I}_n$ is the standard cyclic order \cite{Cachazo:2019ngv}, which can also be viewed as the generalized $\text{tr}(\phi^3)$ amplitude, governs the stringy integral to leading order, as seen in the limit $m^{(k)}_n(\mathfrak{s}) = \lim_{\alpha' \rightarrow 0} I^{(k)}_n(\mathfrak{s}).$
These two limits, giving $m^{(2)}_n$ and $m^{(k)}_n$ for general $k\ge 3$ are computed using respectively the CHY \cite{Cachazo:2013hca} and CEGM formulas \cite{Cachazo:2019ngv},
$$m^{(k)}_n = \sum_{c\in \text{crit}(\mathcal{S}^{(k)}_n)}\frac{1}{\det'H}\left(\frac{1}{p_{12\cdots k}\cdots p_{n12\cdots k-1}}\right)^2,$$
where the sum is over all critical points of the scattering potential $\mathcal{S}^{(k)}_n = \log\left(\prod_{J}p_J^{\mathfrak{s}_J}\right)$
and $\det'H$ is the Hessian determinant of $\mathcal{S}^{(k)}_n$.
	\subsection{The CEGM Nonlinear Sigma Model}
The kinematic deformation of $m^{(2)}_n$ discovered in \cite{Arkani-Hamed:2024nhp} is (assuming $n$ is even) given by 
\begin{eqnarray}\label{eq: delta deformation 2n}
	X_{2i,2j} \mapsto X_{2i,2j} - \delta \text{ and } X_{2i-1,2j-1} \mapsto X_{2i-1,2j-1} + \delta.
\end{eqnarray}
The leading order in the $\delta$-deformation of the biadjoint scalar amplitude $m^{(2)}_n$ is the $n$-point NLSM amplitude.
So for example,
\begin{small}
	$$m^{(2)}_4 = \frac{1}{X_{13}} + \frac{1}{X_{24}} \mapsto \frac{1}{X_{13}+\delta} + \frac{1}{X_{24} - \delta}= -\frac{X_{13} + X_{24}}{\delta^2} + \mathcal{O}(\delta^{-3}),$$
\end{small}
where the coefficient of $\delta^{-2}$ is the 4-point NLSM amplitude.  In this work, we propose a generalization of the NLSM to the CEGM world, from which the NLSM emerges on certain residues.  Let us rewrite Equation \eqref{eq: delta deformation 2n} as $X_{ij} \mapsto X_{ij} + \delta (\vert ij\cap 135\cdots \vert -1)$, i.e., we subtract 1 from the size of the intersection of $\{i,j\}$ with $\{1,3,5,\ldots, \}$.  Just as the NLSM requires that $n$ be divisible by $2$, the \textit{pure} CEGM NLSM requires that $\gcd(k,n)>1$ for nontriviality. 

As we show in Section \ref{sec: two formulas for k propagators}, inverse propagators for CEGM amplitudes $m^{(k)}_n$ possess a natural planar basis \cite{Early:2019eun,Early:2020hap} indexed by $k$-element subsets $X_{j_1\cdots j_k}$, and deformations are particularly elegant in this basis.  The first new CEGM NLSM which we encounter is given by the leading order contribution as $\delta\rightarrow \infty$ of $m^{(3)}_6$, as in Equation \eqref{eq: m36 Intro}, after the generic kinematic shift $\sigma$ defined by 
\begin{eqnarray}\label{eq: generic deformation 36}
	X_{ijk} \mapsto X_{ijk} +\delta(\vert ijk \cap 14 \vert - \vert ijk \cap 25 \vert).
\end{eqnarray}
Applying this to $m^{(3)}_6$ in Equation \eqref{eq: m36 Intro} gives rise to our first CEGM nonlinear sigma model amplitude, where every residue is a triple product of 4-point NLSM amplitudes:
\begin{eqnarray}\label{eq: generic GNLSM Intro A}
	\mathcal{A}^{(3),\sigma}_{6}& = & X_{125} X_{134}+X_{146} X_{134}+X_{124} X_{145}+X_{136} X_{145}+X_{136} X_{235}+X_{145} X_{235}\nonumber\\
	& +& X_{146} X_{236}+X_{245} X_{236}+X_{125} X_{245}+X_{146} X_{245}+X_{124} X_{256}+X_{235} X_{256}\nonumber\\
	& + & X_{124} X_{346}+X_{136} X_{346}+X_{256} X_{346}+X_{125} X_{356}+X_{134} X_{356}+X_{236} X_{356}\nonumber\\
	& -& \frac{\left(X_{124}+X_{236}\right) \left(X_{146}+X_{256}\right) \left(X_{245}+X_{346}\right)}{X_{246}}\nonumber\\
	& - & \frac{\left(X_{125}+X_{136}\right) \left(X_{134}+X_{235}\right) \left(X_{145}+X_{356}\right)}{X_{135}}\\
	& - & \frac{\left(X_{134}+X_{256}\right) \left(X_{125}+X_{346}\right) \left(X_{124}+X_{356}\right)}{X'_{246}}\nonumber\\
	& - & \frac{\left(X_{146}+X_{235}\right) \left(X_{145}+X_{236}\right) \left(X_{136}+X_{245}\right)}{X'_{135}}.\nonumber
\end{eqnarray}
In this work, we explore such objects systematically.  In particular, we shall see that all soft (and hard) limits of $\mathcal{A}^{(3),\sigma}_{6}$ are identically zero.

Returning to the $k=2$ story for a moment shows that the structure of Equation \eqref{eq: generic GNLSM Intro A} is visibly not entirely unfamiliar.  Applying Equation \eqref{eq: delta deformation 2n} to 
\begin{eqnarray*}
	m^{(2)}_6 & = &\frac{1}{X_{14} X_{15} X_{24}}+\frac{1}{X_{15} X_{24} X_{25}}+\frac{1}{X_{24} X_{25} X_{26}}+\frac{1}{X_{13} X_{15} X_{35}}+\frac{1}{X_{15} X_{25} X_{35}}+\frac{1}{X_{13} X_{14} X_{15}}\\
	& + &\frac{1}{X_{13} X_{35} X_{36}}+\frac{1}{X_{26} X_{35} X_{36}}+\frac{1}{X_{13} X_{14} X_{46}}+\frac{1}{X_{14} X_{24} X_{46}}+\frac{1}{X_{24} X_{26} X_{46}}+\frac{1}{X_{25} X_{26} X_{35}}\\
	& + &  \frac{1}{X_{26} X_{36} X_{46}}+\frac{1}{X_{13} X_{36} X_{46}}
\end{eqnarray*}
gives the $n=6$ point NLSM amplitude as the coefficient of $\delta^{-4}$ as $\delta \rightarrow\infty$:
\begin{eqnarray}\label{eq: 6-point NLSM intro}
	\mathcal{A}^{NLSM}_6 & = &  -(X_{13}+X_{15}+X_{24}+X_{26}+X_{35}+X_{46})\nonumber\\
	& + & \frac{\left(X_{15}+X_{26}\right) \left(X_{24}+X_{35}\right)}{X_{25}}+\frac{\left(X_{13}+X_{24}\right) \left(X_{15}+X_{46}\right)}{X_{14}}\\
	& + & \frac{\left(X_{13}+X_{26}\right) \left(X_{35}+X_{46}\right)}{X_{36}}.\nonumber
\end{eqnarray}

Our goal is to generalize this construction to CEGM amplitudes $m^{(k)}_n$.  The first new case is $m^{(3)}_6$, given by the formula
\begin{eqnarray}\label{eq: m36 Intro}
	m^{(3)}_6 = \begin{array}{c}
		\frac{X_{135}+X'_{135}}{X_{135} X_{136} X_{145} X'_{135} X_{235}}+\frac{1}{X_{236} X_{245} X_{246} X_{256}}+\frac{1}{X_{236} X_{246} X_{256} X_{346}}+\frac{X_{246}+X'_{246}}{X_{124} X_{246} X_{256} X'_{246} X_{346}} \\
		+\frac{1}{X_{146} X'_{135} X_{236} X_{245}}+\frac{1}{X_{146} X_{236} X_{245} X_{246}}+\frac{1}{X_{146} X_{236} X_{246} X_{346}}+\frac{1}{X_{145} X_{146} X'_{135} X_{245}} \\
		+\frac{1}{X_{136} X_{146} X'_{135} X_{236}}+\frac{1}{X_{134} X_{136} X_{146} X_{346}}+\frac{1}{X_{136} X_{146} X_{236} X_{346}}+\frac{1}{X_{136} X_{145} X_{146} X'_{135}} \\
		+\frac{1}{X_{125} X_{134} X'_{246} X_{356}}+\frac{1}{X_{125} X_{256} X'_{246} X_{356}}+\frac{1}{X_{136} X_{236} X_{346} X_{356}}+\frac{1}{X_{134} X_{136} X_{145} X_{146}} \\
		+\frac{1}{X_{236} X_{256} X_{346} X_{356}}+\frac{1}{X_{134} X'_{246} X_{346} X_{356}}+\frac{1}{X_{256} X'_{246} X_{346} X_{356}}+\frac{1}{X_{134} X_{136} X_{346} X_{356}} \\
		+\frac{1}{X_{136} X_{235} X_{236} X_{356}}+\frac{1}{X_{125} X_{235} X_{256} X_{356}}+\frac{1}{X_{235} X_{236} X_{256} X_{356}}+\frac{1}{X_{136} X'_{135} X_{235} X_{236}} \\
		+\frac{1}{X'_{135} X_{235} X_{236} X_{245}}+\frac{1}{X_{125} X_{235} X_{245} X_{256}}+\frac{1}{X_{235} X_{236} X_{245} X_{256}}+\frac{1}{X_{145} X'_{135} X_{235} X_{245}} \\
		+\frac{1}{X_{134} X_{135} X_{136} X_{356}}+\frac{1}{X_{125} X_{135} X_{235} X_{356}}+\frac{1}{X_{135} X_{136} X_{235} X_{356}}+\frac{1}{X_{125} X_{145} X_{235} X_{245}} \\
		+\frac{1}{X_{134} X_{135} X_{136} X_{145}}+\frac{1}{X_{125} X_{135} X_{145} X_{235}}+\frac{1}{X_{125} X_{134} X_{135} X_{356}}+\frac{1}{X_{125} X_{134} X_{135} X_{145}} \\
		+\frac{1}{X_{124} X_{245} X_{246} X_{256}}+\frac{1}{X_{124} X_{146} X_{246} X_{346}}+\frac{1}{X_{124} X_{134} X'_{246} X_{346}}+\frac{1}{X_{124} X_{146} X_{245} X_{246}} \\
		+\frac{1}{X_{124} X_{145} X_{146} X_{245}}+\frac{1}{X_{124} X_{125} X_{256} X'_{246}}+\frac{1}{X_{124} X_{134} X_{146} X_{346}}+\frac{1}{X_{124} X_{134} X_{145} X_{146}} \\
		+\frac{1}{X_{124} X_{125} X_{145} X_{245}}+\frac{1}{X_{124} X_{125} X_{245} X_{256}}+\frac{1}{X_{124} X_{125} X_{134} X'_{246}}+\frac{1}{X_{124} X_{125} X_{134} X_{145}} \\
	\end{array}
\end{eqnarray}
where the inverse 3-propagators satisfy $X_{135}+X'_{135} = X_{235} + X_{145} + X_{136}$ and $X_{246} + X'_{246} = X_{124} + X_{256} + X_{346}$.  
	\subsection{Structure of the Paper}

This paper is structured as follows.  Section \ref{sec: geometric deformation theory} introduces the main objects and establishes conventions, and proves that the dimension of the space of pure kinematic shifts is equal to $\gcd(k,n)-1$.  Section \ref{sec: two formulas for k propagators} reviews key properties of the $X$ variables, including a second formula for $X_J$ that aligns with the standard indexing for inverse propagators using bipartitions of $\{1,\ldots, n\}$, and a new proof of their linear independence.
Section \ref{sec: mixed} proves the identification of the $n$-point NLSM amplitude with a certain residue of a mixed $(3,n+2)$ CEGM NLSM amplitude.  Section \ref{sec: hard and soft limits} shows that the pure $(3,6)$ case satisfies an analog of the Adler zero: all soft (and hard) limits vanish identically and conjectures a generalization to all $(3,n)$ when $3$ divides $n$.  Section \ref{sec: discussions} concludes with a summary of results and prospects for future research.

\section{Deformation Theory}\label{sec: geometric deformation theory}

Let us begin by establishing our notation and basic framework.  Let $\lbrack n\rbrack = \{1,\ldots, n\}$.  Denote by $\binom{\lbrack n\rbrack}{k}$ the set of $k$-element subsets of $\lbrack n\rbrack$, and by $\binom{\lbrack n\rbrack}{k}^{ncyc}$ those $k$-element subsets which are not cyclic intervals $\{j,j+1,\ldots, j+k-1\}$.  The \textit{kinematic space} is the subspace $\mathcal{K}_{k,n}$ of $\mathbb{R}^{\binom{n}{k}}$ cut out by the $n$ momentum conservation equations 
$$\sum_{J: J\ni j}\mathfrak{s}_J=0$$
for each $j=1,\ldots, n$, where $\mathfrak{s}_{J}$ is a symmetric tensor with $\mathfrak{s}_{j_1,j_1,\ldots, j_k}=0$ whenever an index is repeated.  It is known (and easy to show by using combinatorial arguments) that momentum conservation is equivalent to imposing that all cyclic $X$-variables, using the formula in Equation \eqref{eq: XJ definition}, vanish identically, i.e. $X_{\lbrack j,j+k-1\rbrack}=0$ for all $j=1,\ldots, n$.

The moduli space $X(k,n)$ is the torus quotient of the open Grassmannian, where all $k\times k$ minors are nonzero.  It has dimension $(k-1)(n-k-1)$, and can also be regarded as the moduli space of $n$ points in $\mathbb{P}^{k-1}$ in general position.  The \textit{positive} configuration space $X^+(k,n)$, first defined in \cite{Arkani-Hamed:2019rds}, is the quotient by $\mathbb{R}_{>0}^n$ of the positive Grassmannian, the subset of the real Grassmannian where all maximal $k\times k$ minors are positive.  
The \textit{positive tropical Grassmannian} \cite{MR2164397,MR2071813} $\text{Trop}^+G(k,n)$ is the chirotropicalization of the real Grassmannian with respect to the all-plus chirotope; for general chirotropical Grassmannians, see \cite{cachazo2024color,Cachazo:2023ltw,Antolini:2024qfr}.  Explicitly, thanks to \cite{Arkani-Hamed:2020cig,MR4241765}, it is the polyhedral fan in $\mathbb{R}^{\binom{n}{k}}$ characterized by the 3-term positive tropical Pl\"{u}cker relations
\begin{eqnarray}\label{eq: positive tropical Plucker relations}
	\pi_{Lac} + \pi_{Lbd} = \min(\pi_{Lab} + \pi_{Lcd},\pi_{Lad} + \pi_{Lbc})
\end{eqnarray}
for each pair $(L,\{a,b,c,d\})$ such that $L \cup \{a,b,c,d\} \in \binom{\lbrack n\rbrack}{k+2}$, with $a<b<c<d$.  But, it contains a large subspace, the lineality space $\text{Lin}_{k,n}$, consisting of all points $\pi \in \mathbb{R}^{\binom{n}{k}}$ with coordinates $\pi_J = \sum_{j\in J}x_j$.  

Finally, denote by $\text{Trop}^+X(k,n) = \text{Trop}^+G(k,n)\slash \text{Lin}_{k,n}$ the \textit{positive tropical moduli space}.  Note that modding out by the lineality space is dual to imposing momentum conservation.  There are many equivalent parameterizations of the positive Grassmannian and the positive configuration space; in this work, we fix the one in Appendix \ref{sec:positive parametrization}.  

Our main question in the study of linear zero preserving deformations has the following simple combinatorial formulation: our task is to compute all possible ways to partition the $n$ cyclic Pl\"{u}cker coordinates
$$p_{12\cdots k},p_{2\cdots k+1},\ldots, p_{n12\cdots k-1}$$
into two sets, such that the ratio of their products is well-defined on the moduli space $X(k,n)$, that is, the factor should be have torus weight zero in each index.  
Let us first review the case $k=2$.  Here there is a single generator for the space of such partitions, defined when $n$ is even:
$$\frac{p_{23}p_{45}\cdots p_{n1}}{p_{12}p_{34}\cdots p_{n-1,n}}.$$
Next, we come to the case $k=3$, in which case we find a two-dimensional space of partitions exactly when $n$ is a multiple of $3$; for example, when $n=6$ it is generated multiplicatively by 
$$\frac{p_{234}p_{156}}{p_{123}p_{456}} \text{ and } \frac{p_{345}p_{126}}{p_{234}p_{156}}.$$
Note that neither of these expressions contain \textit{all} six cyclic minors $p_{i,i+1,i+2}$, but we can construct at least one that does:
$$\left(\frac{p_{234}p_{156}}{p_{123}p_{456}}\right)^2\left(\frac{p_{345}p_{126}}{p_{234}p_{156}}\right) = \frac{p_{234}p_{345}p_{156}p_{126}}{(p_{123}p_{456})^2}.$$
The most general case is when the exponents are generic complex numbers, in which case we can write
$$(p_{123}p_{456})^{z_1}(p_{234}p_{156})^{z_2}(p_{126}p_{345})^{z_3}$$
where the only (linear) relation satisfied by $z_1,z_2,z_3$ is that their sum is zero.  We associate to each such monomial the point in $\sigma(z_1,z_2,z_3)\in \mathcal{K}_{3,6}$ given by 
$$\sigma(z_1,z_2,z_3) = \sum_{j=1}^3 z_j\left(e^{j,j+1,j+2}+e^{j+3,j+4,j+5}\right),$$
i.e. $\mathfrak{s}_{123}=\mathfrak{s}_{456} = z_1$, etc.  We refer to factors involving non-generic (possibly zero) exponents as \textit{special}, and those with parameters such that every cyclic Mandelstam invariant of $\mathfrak{s}_{\lbrack j,j+1\rbrack}$ is nonzero \textit{generic}.  We shall call kinematic shifts where only (a possibly proper subset of) the cyclically consecutive Mandelstams appear \textit{pure}; on the other hand, ratios involving some minors labeled by non-cyclically consecutive subsets may be called \textit{mixed}, as they may give rise to mixed amplitudes as deformations of the biadjoint scalar amplitude $m^{(2)}_n$.  It is exactly the pure kinematic shifts which preserve the (linear) zeros of $m^{(k)}_n$, obtained by setting to zero some of the non-cyclic $\mathfrak{s}_J$'s.

We return to mixed amplitudes in Section \ref{sec: mixed}.

\subsection{Characterization of Pure Kinematic Shifts}
\begin{defn}\label{def: pure kinematic shifts}
	The space of \textit{pure} kinematic shifts is the intersection of $\mathcal{K}_{k,n}$ with the subspace $\text{span}\{e^{\lbrack j,j+k-1\rbrack}:j=1,\ldots, n\} \subset \mathbb{R}^{\binom{n}{k}}$.  
\end{defn}
Our first result is to compute the dimension.  This can be deduced from the dimension of the nullspace of a circulant matrix \cite{ingleton1956rank}; we provide a self-contained proof. 
\begin{prop}
	The dimension of the space of pure kinematic shifts is $\gcd(k,n)-1$.  
\end{prop}
\begin{proof}
	The result is a special case of a general property of circulant matrices as in \cite{ingleton1956rank}; we give a self-contained proof.  The restriction of momentum conservation to the subspace $\text{span}\{e^{\lbrack j,j+k-1\rbrack}:j=1,\ldots, n\}$ reduces to solving the system of $n$ equations 
	\begin{eqnarray*}
		x_1+x_{2} + \cdots +x_{k} & = & 0\\
		x_2+x_{3} + \cdots +x_{k+1} & = & 0\\
		& \vdots & \\
		x_n+x_1 + \cdots +x_{k-1} & = & 0.
	\end{eqnarray*}
	where we define $x_j = \mathfrak{s}_{j,j+1,\ldots, j+k-1}$.  Taking the difference of adjacent equations we find that $x_j = x_{j+k}$ for all $j=1,\ldots, n$, where the indices are cyclic modulo $n$.  Letting $g = \gcd(k,n)$, it follows that these equalities collapse to $x_j = x_{j+g}$ and we have one linear equation which is satisfied among the $g$ variables $x_1,\ldots, x_{g}$, whence the dimension is $g-1$.
\end{proof}
There is now a clear basis of $g-1$ vectors with entries $-1,0,1$ for the space of pure kinematic shifts, 

\begin{eqnarray}\label{eq: pure kinematic shift A}
	\sum_{j=0}^{n/g-1} \left(e^{\lbrack i+gj,i+gj+k-1\rbrack} - e^{\lbrack i+gj+1,i+gj+k\rbrack}\right): i=1,\ldots, g-1.
\end{eqnarray}
For example, when for $(k,n) = (2,4)$ we have the single basis element $e^{12} - e^{23} + e^{34} - e^{14}$, while for $(k,n) = (3,6)$ there are two basis elements
$$e^{123} - e^{234} + e^{456} - e^{156} \text{ and } e^{234} - e^{345} + e^{156} - e^{126}.$$ 

We now present a generating set of solutions for the space of pure kinematic shifts, in terms of the $X_J$ basis (see Theorem \ref{prop: sJ expansion X basis} and in particular Corollary \ref{cor: basis} for the proof of the basis property).  

Denote $g = \gcd(k,n)$.  For each $i=1,\ldots, g$, define a kinematic shift in the direction $\sigma_i \in \mathcal{K}_{k,n}$ by 
\begin{eqnarray}
	X_{J} \mapsto X_J + n_{i,J}\delta,
\end{eqnarray}
where 
$$n_{i,J} = \vert J \cap \{i,i+g,i+2g,\ldots\}\vert -k/g.$$
First note that since the $X_J$ form a basis of the dual kinematic space, our formula uniquely determines the deformation; in other words, $\sigma_i \in \mathcal{K}_{k,n}$ is determined by $X_J = \vert J \cap \{i,i+g,i+2g,\ldots\}\vert -k/g$ for all $J \in \binom{\lbrack n\rbrack}{k}^{ncyc}$.
Clearly, they satisfy the unique relation $\sum_{i=1}^{g}n_{i,J} = 0$ for all $J$.  We claim that all such kinematic shifts are pure, only the cyclic $\mathfrak{s}_J$'s are nonzero.  Let us show this.  Using Appendix \ref{sec: two formulas for k propagators} it follows that 
$$\mathfrak{s}_{j+1,\ldots, j+k} = X_{j,j+2,\ldots, j+k},$$
and consequently it is easy to verify that $n_{i,\{j-1\}\cup \lbrack j+1,j+k\rbrack} = \delta_{j-1,i} - \delta_{j,i+1}$ for each $j=1,\ldots, g$.  We include the full proof that the deformation of Equation \eqref{eq: pure kinematic shift A} is pure in Section \ref{sec: two formulas for k propagators}.

For example, in $\mathcal{K}_{3,6}$ where $g= 3$, only the first one of the following is changed by pure deformations:
\begin{eqnarray*}
	\mathfrak{s}_{234} & = & X_{134},\\
	\mathfrak{s}_{134} & = & X_{124}-X_{134}-X_{246}+X_{346},\\
	\mathfrak{s}_{124} & = & X_{123}-X_{124}-X_{236}+X_{246},
\end{eqnarray*}
noting that $X_{123}=0$, while in $\mathcal{K}_{3,9}$ we have
$$\mathfrak{s}_{135} = -X_{124}+X_{125}+X_{134}-X_{135}+X_{249}-X_{259}-X_{349}+X_{359},$$
which is again not changed by any pure deformation. 

Let us now analyze the deformation of the $X$-basis in the $6$-point case.  The system of 14 equations $\mathfrak{s}_{ijk}=0$ for all $\{i,j,k\} \in \binom{\lbrack n\rbrack}{3}^{ncyc}$ has a two-dimensional solution space.  One solution is given by the kinematic shift 
$$\sigma_1 = e^{234} + e^{156} - e^{123} - e^{456} \in \mathcal{K}_{3,6},$$
or in the $X$ basis
\begin{eqnarray*}
	& & X _{124} \mapsto X _{124} +\delta,\ X _{134} \mapsto X _{134} +\delta,\ X _{145} \mapsto X _{145} +\delta,\ X _{146} \mapsto X _{146} +\delta \\
	& & X _{235} \mapsto X _{235}-\delta,\  X _{236} \mapsto X _{236}-\delta,\ X _{256} \mapsto X _{256}-\delta,\ X _{356} \mapsto X _{356}-\delta.
\end{eqnarray*}
The kinematic shifts $\sigma_1,\sigma_2$ generate the two-dimensional solution space. 

We now introduce the main object of our study.
\begin{defn}
	Given any pure kinematic shift $X_J \mapsto X_J + n_J\delta$ in a given direction  $\sigma \in \mathcal{K}_{k,n}$, denote by $(m^{(k)}_n)_\sigma$ the generalized biadjoint scalar amplitude shifted by $\sigma$.  The generalized NLSM amplitude $\mathcal{A}^{(k),\sigma}_n$ is the leading coefficient of $\delta$ in the $\delta\rightarrow\infty$ expansion of $(m^{(k)}_n)_\sigma$.
\end{defn}

One would like a direct way to compute such deformations using the Global Schwinger Parameterization, see Equation \eqref{eq: GSP0}, which uses a parameterization of the positive tropical Grassmannian to glue together all tree-level Schwinger parameterizations into a single piecewise-exponential integral; it has appeared in numerous places in recent years.  Unfortunately, naively introducing a $\delta$-deformation leads to non-convergent integral.  Luckily, there is a beautiful resolution for which we thank Freddy Cachazo; we introduce the object here, though a detailed treatment is beyond the scope of this paper.  
Recall \cite{Cachazo:2022vuo} that the Global Schwinger Parameterization for $m^{(k)}_n$ is the integral 
\begin{eqnarray}\label{eq: GSP0}
	m^{(k)}_n(\mathfrak{s}_J) & = & \int_{\mathbb{T}^{n-k} \times \cdots \times \mathbb{T}^{n-k}}\exp(-\mathcal{F}^{(k)}_n(\mathfrak{s}, y))dy,
\end{eqnarray}
where $\mathbb{T}^{n-k} = \mathbb{R}^{n-k}\slash \mathbb{R}(1,\ldots, 1)$ is the tropical torus.  Here the tropical potential $\mathcal{F}^{(k)}_n$ is given by 
\begin{eqnarray}\label{eq: tropPot0}
	\mathcal{F}^{(k)}_n(\mathfrak{s},y) & = & \sum_{J \in \binom{\lbrack n\rbrack}{k}}p^{\text{tr}}_J(y) \mathfrak{s}_J,
\end{eqnarray}
where $p^{\text{tr}}_J(y)$ is the tropicalization of the Pl\"{u}cker coordinate $p_J$ in a positive parameterization; we choose the positive parameterization in Appendix \ref{sec:positive parametrization}.

For any (pure) deformation direction
$$\sigma = \sum_{j=1}^n c_j e^{\lbrack j,j+k-1\rbrack}\in \mathcal{K}_{k,n},$$
denote by $\mathcal{A}^{(k),\sigma}_n$ the generalized NLSM amplitude, that is the leading order coefficient where $\delta\rightarrow\infty$ in the integral

\begin{eqnarray}
	(m^{(k)}_n)_\sigma & = & \int_{\mathbb{R}^{(k-1)(n-k-1)}}\exp\left(-\mathcal{F}^{(k)}_n(y) -i\delta \sum_{j=1}^n c_j(p^\text{tr}_{\lbrack j,j+k-1}(y))\right)dy,
\end{eqnarray}
where $i^2 = -1$.  
For example, when $k=2$ and $n=4$, and $\sigma = e^{14} + e^{23} - e^{12} - e^{34}$ we have
\begin{eqnarray}
	(m^{(2)}_4)_{\sigma} & = & \int_{\mathbb{R}}\exp\left(-\mathcal{F}^{(2)}_4 - i\delta(p^\text{tr}_{14}+p^\text{tr}_{23}-p^\text{tr}_{12}-p^\text{tr}_{34})\right)dy
\end{eqnarray}
where
$$\mathcal{F}^{(2)}_4 = \sum_{ij}p^\text{tr}_{ij}s_{ij}$$
and 
$$p^\text{tr}_{12} = 0,\ p^\text{tr}_{13} = 0,\ p^\text{tr}_{14} = 0,\ \ p^\text{tr}_{23} = y_1,\ p^\text{tr}_{24} = \min(y_1,y_2),\ p^\text{tr}_{34} = y_2.$$
Evaluating the integral after choosing the gauge where $y_1=0$, say, gives
\begin{eqnarray}
	(m^{(2)}_4)_\sigma	& = & \frac{1}{X_{13} - i\delta} + \frac{1}{X_{24} + i\delta},
\end{eqnarray}
and the leading order coefficient when $\delta \rightarrow\infty $ is the NLSM amplitude $\mathcal{A}^{(2),\sigma}_4 = X_{13}+X_{24}$:
$$(m^{(2)}_4)_\sigma = \frac{X_{13}+X_{24}}{\delta ^2}-\frac{i \left(X_{13}-X_{24}\right) \left(X_{13}+X_{24}\right)}{\delta ^3}+\mathcal{O}(\delta^{-4}).$$
In the rest of this paper, we suppress the imaginary unit $i$ from the deformation factor.
\section{$k$-Propagators: Formulas and Linear Independence}\label{sec: two formulas for k propagators}
Much of the material in this section has appeared before, including in \cite{Cachazo:2020uup,Cachazo:2020wgu,Early:2019eun,Early:2020hap,Early:2021tce,Early:2022zny,Early:2022mdn,Early:2018zuw}; we gather it here in an improved exposition. The directed distance function $d_{I,J}$ is the minimal number of steps in the directions $e_1-e_2,e_2-e_3, \ldots, e_n-e_1$ required to walk from the vertex $e_I$ to the vertex $e_J$ on the 1-skeleton of $\Delta_{k,n}$.  We will show that this rule can be written very compactly as
\begin{eqnarray}\label{eq: dIJ Li}
	d_{I,J} & = & -\min\{L_1(e_J-e_I),\ldots, L_n(e_J- e_I)\}
\end{eqnarray}
where
$$L_j(x) = x_{j+1} + 2x_{j+2} + \cdots + (n-1)x_{j-1}.$$
With this preparation, we can define our main combinatorial object, the planar basis element $X_J$, by	
\begin{eqnarray}\label{eq: XJ definition}
	X_J(\mathfrak{s}) & = & \frac{1}{n}\sum_{I\in \binom{\lbrack n\rbrack}{k}}d_{J,I}\mathfrak{s}_I.
\end{eqnarray}
\begin{figure}[h!]
	\centering
	\includegraphics[width=0.6\linewidth]{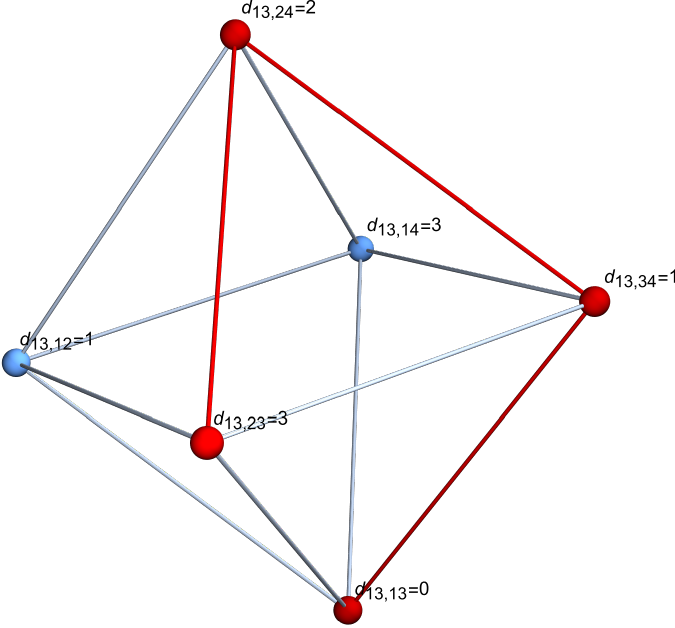}
	\caption{Directed distance function computation of $X_{13} = \frac{1}{4} \left(s_{12}+3 s_{14}+3 s_{23}+2 s_{24}+s_{34}\right)$.  In this case, it simplifies modulo momentum conservation to just $s_{23}$.  The length 3 path from $e_{1}+e_3$ to $e_{2}+e_3$ is shown in red, where we abbreviate $e_i+e_j \in \Delta_{2,4}$ as $ij$.}
	\label{fig:distanceoctahedron}
\end{figure}
\begin{example}
	The expansions of $X_{13}$ and $X_{246}$ are 
	$$X_{13} = \frac{1}{4} \left(s_{12}+3 s_{14}+3 s_{23}+2 s_{24}+s_{34}\right) = s_{23},$$
	and 
	\begin{eqnarray*}
		X_{246} & = & \frac{1}{6}\left(6 s_{123}+5 s_{124}+4 s_{125}+3 s_{126}+4 s_{134}+3 s_{135}+2 s_{136}+2 s_{145}+s_{146}+6 s_{156}\right.\\
		& + & \left. 3 s_{234}+2 s_{235}+s_{236}+s_{245}+5 s_{256}+6 s_{345}+5 s_{346}+4 s_{356}+3 s_{456}\right).
	\end{eqnarray*}
	See Figure \ref{fig:distanceoctahedron}.
\end{example}
Let us now prove Equation \eqref{eq: dIJ Li}, thereby putting on firm footing Equation \eqref{eq: XJ definition} for the $X$-variables.  For each $j=1,\ldots, n$ let $\Pi_j$ be the simplicial cone
$$\Pi_j = \left\{t_1(e_1-e_2)+t_2(e_2-e_3)+\cdots +t_n(e_n-e_1):t_1,\ldots, t_n\ge 0;\ t_{j-1}=0\right\}.$$
Given $u,v \in \mathbb{Z}^n$ with $\sum_{i=1}^n u_i =\sum_{i=1}^n v_i \in\mathbb{Z}$, then, noting that $\Pi_1,\ldots, \Pi_n$ are the maximal cones in a complete simplicial fan, i.e. in the inner normal fan to a certain alcoved simplex, it follows that $u-v$ lies in the relative interior of a maximal intersection of a unique subset of the $\Pi_j$, so that
$$u-v = \sum_{\ell \in \{j_1,\ldots, j_d\}} t_\ell(e_\ell - e_{\ell+1}).$$
Supposing without loss of generality that $t_n=0$, then the formula is $-\min\{L_1(u-v),\ldots, L_n(u- v)\} = -L_1(u-v)$, where
\begin{eqnarray}\label{eq: L plate}
	L_1(u-v) & = & (u-v)\cdot \sum_{j=2}^{n}\left(e_j+e_{j+1}+\cdots e_{n}\right)\\
	& = & \left(\sum_{j=1}^{n-1} t_j (e_j-e_{j+1})\right)\cdot \sum_{j=2}^{n}\left(e_j+e_{j+1}+\cdots +e_{n}\right)\nonumber\\
	& = & -\sum_{j=1}^{n-1}t_j,\nonumber
\end{eqnarray}
whose negative is the total distance walked from $v$ to $u$ in the directions of the cyclic roots $e_1-e_2,e_2-e_3,\ldots, e_n-e_1$.  In particular, when $(u,v) = (e_I,e_J)$ is a pair of vertices of $\Delta_{k,n}$, the formula evaluates to $d_{J,I}=\sum_{j=1}^{n-1}t_j$ and the result holds.

When $k=2$ we are in the case of standard QFT propagators; we now show our formula coincides with an index shift of the more standard expression, $X'_{i,j} := \left(p_{i+1} + \cdots +p_{j}\right)^2 = \sum_{i+1\le a<b\le j}s_{ab}$, where we are taking into account that $p_j^2=0$ as usual.
\begin{prop}\label{prop: planar basis identity 2n}
	For any distinct $i,j$, taking into account momentum conservation we have
	$$X_{i,j} = \left(p_{i+1} + \cdots +p_{j}\right)^2,$$
	where $X_{i,j}$ is given by Equation \eqref{eq: XJ definition}.
\end{prop}

\begin{proof}
	For the computational argument, one simply makes the substitution $s_{ij} = -(X_{i,j} + X_{i+1,j+1} - X_{i+1,j} - X_{i,j+1})$ into the formula for $X_{i,j}$ and simplifies, but this does not provide much insight; we give a geometric argument which provides crucial intuition when $k\ge 3$.  We need to show that $X'_{i,j}$ and $X_{i,j}$ induce the same subdivision of $\Delta_{2,n}$ into two matroid polytopes that share the same internal facet
	\begin{eqnarray}
		x_{i+1} + \cdots + x_j = 1,
	\end{eqnarray}
	and that in particular that subdivision is coarsest; from this it follows that $X'_{i,j}$ and $X_{i,j}$ are proportional.  The normalization factor on $X_{i,j}$ makes them equal.
	
	To whit: we claim that $X'_{i,j}$ and $X_{i,j}$ induce the subdivision of $\Delta_{2,n}$, consisting of the two matroid subpolytopes that share the following common internal facet:
	\begin{eqnarray}\label{eq: 2-split}
		x_{i+1} + \cdots + x_j = 1 (= x_{j+1} + \cdots + x_i),
	\end{eqnarray}
	noting that $x_1+\cdots +x_n=2$.
	
	We first apply momentum conservation to rewrite $X'_{i,j}$:
	\begin{eqnarray}
		X'_{i,j} & = & -(p_{i+1}+ \cdots p_j)\cdot (p_{j+1} + \cdots + p_i)\\
		& = & -\sum_{(a,b) \in \lbrack j+1,i\rbrack \times \lbrack i+1,j\rbrack} s_{a,b},
	\end{eqnarray}
	noting that the coefficients of the Mandelstam invariants $s_{a,b}$ define, modulo lineality, a weight vector over the vertices of $\Delta_{2,n}$ which achieves its minimum (at least) twice on the hyperplane in Equation \eqref{eq: 2-split}.  Consequently it induces the subdivision of $\Delta_{2,n}$ consisting of the two matroid polytopes 
	\begin{eqnarray}\label{eq: facet inequalities 2-split}
		x_{i+1} + \cdots + x_j \ge 1\text{ and } x_{j+1} + \cdots + x_{i}\ge 1,
	\end{eqnarray}
	respectively.  We will show in two steps that $X_{i,j}$ induces the same subdivision as $X'_{i,j}$.  First, we observe that the piecewise linear function
	$$\rho_{ij}(x) := -\min\{L_1(x-(e_i+e_j)),\ldots, L_n(x-(e_i+e_j)))\},$$
	satisfies $\rho_{ij}(e_a+e_b) = d_{ij,ab}$, and that $\rho_{ij}(x)$ has a discontinuous derivative across the Cartesian product of simplices
	\begin{eqnarray}\label{eqn: cartesian product 2n}
		\left\{x\in \Delta_{2,n}: x_{i+1} + \cdots + x_{j} = x_{j+1} + \cdots + x_i=1\right\},
	\end{eqnarray}
	or equivalently
	\begin{eqnarray}
		\left\{x\in \Delta_{2,n}: L_i(x) = L_j(x)\right\}.
	\end{eqnarray}
	This construction subdivides $\Delta_{2,n}$ into the two matroid polytopes in Equation \eqref{eq: facet inequalities 2-split}: $\rho_{ij}$ is linear over each polytope, and its derivative is discontinuous across their common intersection.
	
	Since $X'_{i,j}$ and $X_{i,j}$ induce the same subdivision, and because that subdivision is coarsest, i.e. it refines only the trivial subdivision, that means that they must be proportional.  That proportionality constant is fixed by $\frac{1}{n}$ in the formula for $X_{i,j}$.
\end{proof}

Now let $(\mathbf{S},\mathbf{r})$ be a decorated ordered set partition of type $\Delta_{k,n}$, where $\mathbf{S} = (S_1,\ldots, S_d)$.  Define
$$X_{(\mathbf{S},\mathbf{r})} = -\frac{1}{d}\sum_{J \in \binom{\lbrack n\rbrack}{k}}\min\left(M_{1}(e_J),\ldots, M_{d}(e_J)\right),$$
where
$$M_j(x) = r_{j+1}x_{S_{j+1}} + (r_{j+1}+r_{j+2})x_{S_{j+2}} + \cdots + (r_{j+1}+\cdots +r_{j-1})x_{S_{j-1}}.$$
having denoted $x_{S_j} = \sum_{i\in S_j}x_i$.

\begin{prop}[\cite{Early:2022mdn}]
	Given any $J \in \binom{\lbrack n\rbrack}{k}$, Let $(\mathbf{S},\mathbf{r})$ be the decorated ordered set partition defined as follows.  Let $j_{\ell_1},\ldots, j_{\ell_d}$ be the endpoints of $J$ with respect to the cyclic order $(1,2,\ldots, n)$.  For each $i=1,\ldots, d$, let $S_{\ell_i}$ be the cyclic interval
	$$S_{\ell_i} = \lbrack j_{\ell_{i-1}+1},j_{\ell_{i}} \rbrack$$
	and let 
	$$r_{\ell_i} = \big\vert \lbrack j_{\ell_{i-1}}+1,j_{\ell_{i}}\rbrack \cap J \big\vert.$$
	We have
	$$X_{J} = X_{(\mathbf{S},\mathbf{r})}.$$
\end{prop}
We give the main idea of the proof from \cite{Early:2022mdn}.
\begin{proof}[Sketch of Proof]
	The proof that $X_J$ and $X_{(\mathbf{S},\mathbf{r})}$ induce the same (coarsest) subdivision can be easily deduced from \cite[Theorem 17]{Early:2022zny}.  The task is to verify, by computing where the minima in each case is achieved twice, that the tropical hyperplane defined by 
	$$\rho_{(\mathbf{S},\mathbf{r})}(x) = -\frac{1}{d}\min(M_{S_1}(x),\ldots, M_{S_d}(x))$$
	coincides (on their intersection with $\Delta_{k,n}$) with the tropical hyperplane defined by 
	$$\rho_{J}(x) = -\frac{1}{n}\min\left(L_1(x-e_J),\ldots, X_n(x-e_J)\right),$$
	where 
	$$L_i = x_{i+1} + 2 x_{i+2} + \cdots + (n-1)x_{i-1}.$$
\end{proof}
In the next section, we continue by giving the expansion of $\mathfrak{s}_J$ in the $X$-basis.  If $J = (j_1<\cdots <j_k) \in \binom{\lbrack n\rbrack}{k}^{ncyc}$, denote by $I_J$ the cyclic \textit{initial} indices of $J$:
$$I_J = \left\{j \in \lbrack n\rbrack: (j-1,j) \in J^c \times J\right\}.$$

\subsection{Basis Proof}
The following Theorem was proved first in \cite{Early:2020hap}, but the argument was quite short; therefore we present a new proof with more details.  The important point is that it proves that the $X$-variables are linearly independent, and form a basis of linear functions on the kinematic space.
\begin{thm}[\cite{Early:2020hap}]\label{prop: sJ expansion X basis}
	We have
	\begin{eqnarray}\label{eq: mandelstam expand}
		\mathfrak{s}_J & = & \sum_{e_M \in C_J}(-1)^{\vert J \cap M\vert-k-1}X_{M}
	\end{eqnarray}
	where the sum is over the cubical array
	$$C_J = \left\{e_J + \sum_{M \subseteq I_J}(e_{m_{j}-1} - e_{m_{j}}) \right\},$$
	where $M = (m_1,m_2,\ldots)$.
\end{thm}
In the calculation which follows, we do not assume the full momentum conservation; only the relation $\sum_{J\in \binom{\lbrack n\rbrack}{k}}\mathfrak{s}_J=0$, which is itself is implied by momentum conservation. 
\begin{proof}
	First suppose that $J$ is a single cyclic interval, say  $J= \lbrack 2,k+1\rbrack$.  Denote $J' = \{1,3,\ldots, k+1\}$.  Let $I \in \binom{\lbrack n\rbrack}{k}$ be arbitrary.  Now, since $C_J = \{J,J'\}$ and any path from $J$ to $I$ has to go through $J'$, it follows that the coefficient $d_{J,I} - d_{J',I}$ of $\mathfrak{s}_I$ in $X_J - X_{J'}$ is equal to $1$ unless $I=J$, in which case it is equal to one more than the length of the shortest path from $e_J'$ to $e_J$, that is 
	$$e_{\{1,3,4,\ldots, k+1\}} \rightarrow\cdots  \rightarrow e_{\{1,2,\ldots, k\}} \rightarrow \cdots \rightarrow e_{\{2,3,\ldots, k+1\}},$$
	where the first ellipsis $\cdots$ consists of $k-1$ steps, and the second consists of $n-k-1$ steps.  Thus, 
	$$d_{J,I} - d_{J',I} = -(1+ (k-1) + (n-k-1)) = -(n-1).$$
	It follows that $X_J - X_{J'} = -\mathfrak{s}_J +\frac{1}{n-1} \sum_{I}\mathfrak{s}_I \equiv -\mathfrak{s}_J$, where the last identity holds because the constant height surface over $\Delta_{k,n}$ is in the lineality space, hence is implied by momentum conservation. 
	
	Next let us suppose that $J$ is a disjoint union of at least two cyclic intervals; we compute the coefficient of $\mathfrak{s}_J$ in the right hand side of Equation \eqref{eq: mandelstam expand}.
	
	The (trivial) path from $e_J$ to itself in steps parallel to roots $e_i-e_{i+1}$ has length zero; all others in the sum contributing to the coefficient of $e^J$ are shortenings of the long path (of length $n$) between $e_J$ and itself, so that $d_{J,J} = 0$, and otherwise $d_{J,L} = n - (k-\vert L\cap J\vert)$.  Consequently the alternating sum is now $n$ rather than $0$.  We obtain
	\begin{eqnarray*}
		\left(\sum_{L\in C_J}(-1)^{\vert J \cap L\vert-k-1}X_L\right)_{\lbrack \mathfrak{s}_J\rbrack} & = & \frac{1}{n}\left(\sum_{L \in C_J}(-1)^{\vert J \cap L\vert-k-1}d_{L,J}\right)\\
		& = & \frac{1}{n}\left(\sum_{L \in C_J \setminus J}(-1)^{\vert J \cap L\vert-k-1}d_{L,J}\right)\\
		& = & \frac{1}{n}\left(n + \sum_{L \in C_J}(-1)^{\vert J \cap L\vert-k-1}(n-(k-\vert L\cap J\vert))\right)\\
		& = & 1 + \sum_{L \in C_J}(-1)^{\vert J \cap L\vert-k-1}\vert L\cap J\vert\\
		& = & 1,
	\end{eqnarray*}
	where in the third line we have added back in the $n$ to take into account the contribution when $L = J$, so that $\vert J\cap L\vert = k$, hence  $n-(k - k)=n$.
	
	Now for the case $I\not=J$, where we again assume that $J$ is not a single cyclic interval.  We break this into two steps.  If $I \not\in C_J$, then $d_{L,I} = d_{J,I}- d_{J,L}$, hence 
	\begin{eqnarray*}
		\left(\sum_{L\in C_J}(-1)^{\vert J \cap L\vert-k-1}X_L\right)_{\lbrack \mathfrak{s}_I\rbrack} & = & \frac{1}{n}\left(\sum_{L \in C_J}(-1)^{\vert J \cap L\vert-k-1}d_{L,I}\right)\\
		& = & \frac{1}{n}\left(\sum_{L \in C_J}(-1)^{\vert J \cap L\vert-k-1}(d_{J,I}- d_{J,L})\right)\\
		& = & \frac{1}{n}\left(\sum_{L \in C_J}(-1)^{\vert J \cap L\vert-k-1}( d_{J,I})\right)-\frac{1}{n}\left(\sum_{L \in C_J}(-1)^{\vert J \cap L\vert-k-1}( d_{J,L})\right)\\
		& = & 0,
	\end{eqnarray*}
	since each sum in the third line vanishes separately.  Now suppose that $I \in C_J$.  We have two cases:
	\begin{eqnarray*}
		d_{L,I} & = & \begin{cases} d_{J,I} - d_{J,L}, & d_{J,I} >d_{J,L}\\
			n + (d_{J,I} - d_{J,L}), & d_{J,I} < d_{J,L}. 
		\end{cases}
	\end{eqnarray*}
	and the sum decomposes into two parts that vanish separately:
	\begin{eqnarray*}
		\left(\sum_{L\in C_J}(-1)^{\vert J \cap L\vert-k-1}X_L\right)_{\lbrack \mathfrak{s}_I\rbrack} & = & \frac{1}{n}\left(\sum_{L \in C_J}(-1)^{\vert J \cap L\vert-k-1}d_{L,I}\right)\\
		& = & \frac{1}{n}\left(\sum_{L \in C_J,\ d_{J,I} >d_{J,L}}(-1)^{\vert J \cap L\vert-k-1}(d_{J,I} - d_{J,L})\right)\\
		& + & \frac{1}{n}\left(\sum_{L \in C_J,\ d_{J,I} < d_{J,L}}(-1)^{\vert J \cap L\vert-k-1}(n + (d_{J,I} - d_{J,L}))\right)\\\\
		& = &0.
	\end{eqnarray*}
	
\end{proof}

\begin{cor}\label{cor: basis}
	The set $\left\{X_J: J \in \binom{\lbrack n\rbrack}{k}^{ncyc}\right\}$ is a basis of linear functions on $\mathcal{K}_{k,n}$.
\end{cor}

\begin{example}
	In $\mathbb{R}^{\binom{4}{2}}$ we have 
	\begin{eqnarray*}
		\left(X_{24} - X_{14}-X_{23}+X_{13}\right)_{\lbrack \mathfrak{s}_{12}\rbrack} & = &  \frac{1}{4}(d_{e_{24},e_{12}} - d_{e_{14},e_{12}} - d_{e_{23},e_{12}} + d_{e_{13},e_{12}})\\
		& = & \frac{1}{4}(1-2-2+3)\\
		& = & 0,
	\end{eqnarray*}
	while 
	\begin{eqnarray*}
		\left(X_{24} - X_{14}-X_{23}+X_{13}\right)_{\lbrack \mathfrak{s}_{24}\rbrack} & = & \frac{1}{4}(d_{e_{24},e_{24}} - d_{e_{14},e_{24}} - d_{e_{23},e_{24}} + d_{e_{13},e_{24}})\\
		& = & \frac{1}{4}(0-3-3+2)\\
		& = & -1,
	\end{eqnarray*}
	as expected.
\end{example}

\subsection{Kinematic Deformation via the $X$-Basis}

\begin{prop}\label{prop: propagators XJ deformation}
	The kinematic deformation $\sigma_j\in \mathcal{K}_{k,n}$ defined by 
	$$X_J = \vert J\cap (1,k+1,2k+1,\ldots )\vert -k/g$$
	is pure.
\end{prop}

\begin{proof}
	We invoke the formula from Proposition \ref{prop: sJ expansion X basis}, evaluating
	\begin{eqnarray*}
		\mathfrak{s}_J & = & \sum_{e_M \in C_J}(-1)^{\vert J \cap M\vert-k-1}X_{M}
	\end{eqnarray*}
	at the point $\sigma_i \in \mathcal{K}_{k,n}$, by replacing $X_M$ with $(\vert M\cap \{i,i+g,\ldots, \}\vert - k/g)$, giving 
	\begin{eqnarray*}
		\mathfrak{s}_J(\sigma_i) & = & \sum_{e_M \in C_J}(-1)^{\vert J \cap M\vert-k-1}(\vert M\cap \{i,i+g,\ldots, \}\vert - k/g)
	\end{eqnarray*}
	This expression is an alternating sum over the vertices of a cube of dimension $d$, where $d$ is the number of cyclic intervals in the set $J$.  In fact, this sum cancels on each translation of a given square face of that cube, which reduces us to a cancellation for a case which we know already, for $\mathcal{K}_{2,n}$.  
\end{proof}
For example, for the base case of the proof, letting $\sigma_1 = e^{23} + e^{14} - e^{13} - e^{24} \in \mathcal{K}_{2,4}$, then $s_{24}$ evaluates to 
\begin{eqnarray*}
	s_{24} & = & -X_{24} + X_{14} + X_{23} - X_{13}\\
	& = & -(0-1) + (1-1) + (1-1) - (2-1) = 0.
\end{eqnarray*}
On the other hand, consider the expansion of $\mathfrak{s}_{246}$ at the point $\sigma_1 = e^{234} + e^{156} - e^{123} - e^{456}$.  We evaluate at $\sigma_1$, grouping the terms into two 4-term identities and obtain
\begin{eqnarray*}
	\mathfrak{s}_{246} & = & X_{135}-X_{136}-X_{145}+X_{146}-X_{235}+X_{236}+X_{245}-X_{246}\\
	& = & - (X_{246} - X_{236} - X_{245} + X_{235}) + (X_{146} - X_{136} - X_{145} + X_{135})\\
	& = &- \left((1-1) - (0-1) - (1-1) + (0-1)\right) \\
	& + & \left((2-1) - (1-1) - (2-1) + (1-1)\right)\\
	& =& 0.
\end{eqnarray*}
We conclude this section with Theorem \ref{thm: properties eta}, proved in \cite{Early:2022zny}, which makes possible the equivalence of the following statements:
\begin{enumerate}
	\item $X_I,X_J$ are compatible (in the sense that they can appear together in the same Generalized Feynman Diagram),
	\item Exactly two sign flips in $e_I-e_J$ with respect to the cyclic order $(12\cdots n)$,
	\item No octahedral face of $\Delta_{k,n}$ is triangulated, i.e. the induced subdivision of $\Delta_{k,n}$ induced by $X_I + X_J$ is matroidal. 
\end{enumerate}
\begin{thm}[\cite{Early:2022zny}]\label{thm: properties eta}
	The following hold. 
	\begin{enumerate}
		\item $X_J$ induces a subdivision of $\Delta_{k,n}$ into positroid polytopes; in particular, this subdivision is coarsest, and $X_J$ is in the direction of a ray of $\text{Trop}^+(X(k,n))$.
		\item Let $I,J \in \binom{\lbrack n\rbrack}{k}$ be given.  The subdivision of $\Delta_{k,n}$ induced by any positive linear combination $c_IX_I + c_JX_J$ is the common refinement of the two matroid subdivisions induced by $X_I$ and $X_J$, respectively.  The cells of this subdivision are matroid polytopes if and only if $e_I-e_J$ alternates sign exactly twice with respect to the standard cyclic order $(1,2,\ldots, n)$.  
	\end{enumerate}
\end{thm}
The proof \cite{Early:2022zny} involves an induction on the face lattice of $\Delta_{k,n}$, reducing the computation to the compatibility of standard propagators on faces $\Delta_{2,n-(k-2)}$; it requires a careful accounting of how the octahedral faces of $\Delta_{k,n}$ are subdivided by each $X_J$.  The main technical point is the following: if $e_I-e_J$ flips sign more than twice, produce an octahedral face of $\Delta_{k,n}$ which is subdivided non-matroidally by $X_I + X_J$ (into four tetrahedra).

The compatibility relations in Theorem \ref{thm: properties eta} are a valuable tool for the computation of residues and deformations; in particular, it comes in handy in Section \ref{sec: factorization m48} for the task of reassembling the two copies of $m^{(4)}_6$ (and from that, the 6-point NLSM amplitudes) from their inverse propagators.

\section{Mixed $3$-Amplitudes and a Proof of Life}\label{sec: mixed}
In this section, we assume that $n$ is even and introduce a mixed deformation which reveals an embedding of the $n$-point NLSM into a mixed $(3,n+2)$ CEGM NLSM. Let us begin with the mixed deformation of $m^{(3)}_6$ given by 
$$\sigma = e^{235} + e^{145} + e^{136} - e^{123} - e^{345} - e^{156},$$
or equivalently
$$X_{ijk} \mapsto X_{ijk} + \delta(\vert ijk\cap 135\vert-2).$$
We emphasize that this deformation is \textit{not} zero preserving; as such, it may be expected to give rise to mixed amplitudes.  This is indeed what happens on residues.
The mixed $3$-amplitude is given by 
\begin{eqnarray*}
	\mathcal{A}^{(3),\sigma}_6 & = & \frac{1}{X_{134} X_{136}}+\frac{1}{X_{125} X_{145}}+\frac{1}{X_{134} X_{145}}+\frac{1}{X_{125} X_{235}}+\frac{1}{X_{125} X_{356}}+\frac{1}{X_{125} X_{134}}\\
	& + & \frac{1}{X_{136} X_{356}}+\frac{1}{X_{235} X_{356}}+\frac{1}{X_{134} X_{356}}\\
	& - & \frac{X_{124}+X_{135}}{X_{125} X_{134} X_{145}}-\frac{X_{135}+X_{146}}{X_{134} X_{136} X_{145}}-\frac{X_{135}+X_{245}}{X_{125} X_{145} X_{235}}-\frac{X_{135}+X_{236}}{X_{136} X_{235} X_{356}}\\
	& - &\frac{X_{135}+X_{256}}{X_{125} X_{235} X_{356}}-\frac{X_{135}+X'_{246}}{X_{125} X_{134} X_{356}}-\frac{X_{135}+X_{346}}{X_{134} X_{136} X_{356}}
\end{eqnarray*}
where
$$X'_{246} +X_{246} = X_{124}+X_{256}+X_{346}.$$
We report some interesting behavior on residues.  All first-order residues are naturally identified with mixed $\delta$-deformations of the 6-point biadjoint scalar and are accordingly mixed amplitudes.  For example,
\begin{eqnarray*}
	\text{Res}_{X_{134}=0}\left(\mathcal{A}^{(3),\sigma}_6\right) & = & \frac{1}{X_{136}}+\frac{1}{X_{145}}+\frac{1}{X_{356}}+\frac{1}{X_{125}}\\
	& - & \frac{X'_{246}+X_{135}}{X_{125} X_{356}}-\frac{X_{124}+X_{135}}{X_{125} X_{145}}-\frac{X_{135}+X_{146}}{X_{136} X_{145}}+\frac{X_{135}+X_{346}}{X_{136} X_{356}}
\end{eqnarray*}
After an identification of kinematic invariants, this is the same as the mixed amplitude
$$\frac{1}{X_{26}}+\frac{1}{X_{35}}+\frac{1}{X_{46}}+\frac{1}{X_{13}}-\frac{X_{14}+X_{36}}{X_{13} X_{46}}-\frac{X_{15}+X_{36}}{X_{13} X_{35}}-\frac{X_{25}+X_{36}}{X_{26} X_{35}}-\frac{X_{24}+X_{36}}{X_{26} X_{46}},$$
which arises from the 6-point biadjoint scalar using the deformation
$$\sigma =e^{16}+e^{34} - e^{13} - e^{46}$$
or equivalently
$$X_{ij} \mapsto X_{ij} + \delta(\vert ij \cap 1245\vert-1).$$
All codimension-2 residues are mixed 5-point amplitudes of the same form.  For example,
\begin{eqnarray*}
	\text{Res}_{X_{136}}\left(\text{Res}_{X_{245}}(\mathcal{A}^{(3),\sigma}_6)\right) & = & 1-\frac{X_{135}+X_{146}}{X_{145}}-\frac{X_{135}+X_{346}}{X_{356}}\\
	\text{Res}_{X_{146}}\left(\text{Res}_{X_{245}}(\mathcal{A}^{(3),\sigma}_6)\right) & = & 1-\frac{X'_{246}+X_{135}}{X_{125}}-\frac{X_{135}+X_{346}}{X_{136}},
\end{eqnarray*}
both of which are identified with a 5-point mixed amplitude (compare to \cite[Equation 2.21]{Cachazo:2016njl}), namely
$$A_5(1^\phi,2^\phi,3^\phi,4^\Sigma,5^\Sigma) = 1-\frac{X_{14}+X_{35}}{X_{13}}-\frac{X_{24}+X_{35}}{X_{25}},$$
obtained via the mixed kinematic shift $\sigma \in \mathcal{K}_{2,5}$ given by 
$$\sigma = e^{13} + e^{45} -e^{15} - e^{34}.$$
See also \cite{Low:2018acv}.  Next, we provide a proof of life: we prove that NLSM amplitudes live on residues of certain mixed generalized NLSM amplitudes.  Supposing that $n$ is even, we propose an embedding of the $(n-2)$-particle NLSM amplitude as a residue of the mixed deformation $\sigma$ of $m^{(3)}_n$ given by 
\begin{eqnarray}\label{eqn: Xijk mixed deformation}
	X_{ijk} & \mapsto & X_{ijk} + \delta(\vert 135\cdots n-1\vert -1).
\end{eqnarray}

\subsection{Residual Embeddings}

As a first step, we recall from \cite{Cachazo:2022vuo} the residual embedding of $m^{(2)}_n$ in $m^{(3)}_n$.

\begin{thm}[\cite{Cachazo:2022vuo}]\label{thm: residual embedding}
	The sequence of residues of $m^{(3)}_n$ where 
	\begin{eqnarray}\label{eq: residual embedding}
		X_{23n} = X_{34n} = \cdots = X_{n-4,n-3,n}=0
	\end{eqnarray}
	is equal to $m^{(2)}_n$ after an identification of kinematic invariants.
\end{thm}

The main idea of the proof in \cite{Cachazo:2022vuo} was to construct a fibration over a degenerate space of $n$-gons, where $n-3$ of the vertices are collinear; then the CEGM integral was shown to reduce on the residue prescribed in Equation \eqref{eq: residual embedding} to the usual CHY integral for $m^{(2)}_n$ on $X^+(2,n)$.

This degenerate space of $n$-gons is depicted in Figure \ref{fig:residualembedding3n}, in which  points $1,2,\ldots, n-3$ are on a line $L$, while points $n-2,n-1,n$ are generic.  These three points map to three points on $L$ according to the birational map described in Figure \ref{fig:residualembedding3n}, completing the identification.  This gives a second proof, realizing $X^+(2,n)$ explicitly as a face of $X^+(3,n)$.

\begin{figure}[h!]
	\centering
	\includegraphics[width=0.24\linewidth]{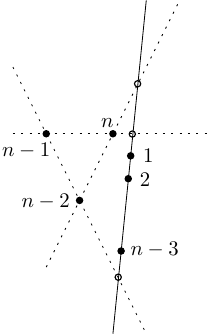}
	\caption{The mechanics of Theorem \ref{thm: residual embedding}.  Evaluating the given residue localizes $n$ points in $\mathbb{P}^2$ to the degenerate configuration where points $1,2,\ldots, n-3$ are collinear in a line $L$.  The remaining $3$ points project in pairs onto $L$, determining a point in $X(2,n) = \mathcal{M}_{0,n}$.}
	\label{fig:residualembedding3n}
\end{figure}

\begin{cor}
	The sequence of residues of $m^{(3)}_n$, where each of the collection of inverse propagators 
	$$X_{23n},X_{34n},X_{45n},\ldots, X_{n-4,n-3,n},X_{n-3,n-2,n}$$
	vanishes, is equal to $m^{(2)}_{n-1}$ after an identification of kinematic invariants.  A natural identification of inverse propagators is $X_{ij} = X_{ijn}$.  The same residue of $\mathcal{A}^{(3),\sigma}_{n}$ is a mixed amplitude on $n-2$ particles $1,2,\ldots, n-2$.
\end{cor}
We need to take one more residue in order to find the NLSM amplitude.
\begin{thm}
	The additional residue where $X_{1,n-2,n}=0$ is identified with the $(n-2)$-point NLSM amplitude.
\end{thm}
\begin{proof}
	We start with the mixed deformation 
	\begin{eqnarray}\label{eq: mixed deformation}
		X_{ijk} & \mapsto & X_{ijk} + \delta(\vert ijk\cap 135\cdots n-1\vert-1)
	\end{eqnarray}
	and restrict to the case at hand, i.e. when $k=n$.  Note that this gives the usual pure NLSM-type deformation, since we have the specialization under $X_{ij} = X_{ijn}$ of Equation \eqref{eq: mixed deformation} to 
	$$X_{ij} \mapsto X_{ij} + \delta(\vert ij\cap 135\cdots n-1\vert -1).$$

	This gives rise to a mixed amplitude consisting of pions and biadjoint scalars.  Noting that under the identification $X_{1,n-2,n} = X_{1,n-2} = s_{1,n-1}$, which is a pole of $m^{(2)}_{n-1}$, it follows that taking the additional residue where $X_{1,n-2,n}=0$ results in an $(n-2)$-point NLSM amplitude. 
\end{proof}
\begin{figure}[h!]
	\centering
	\includegraphics[width=0.6\linewidth]{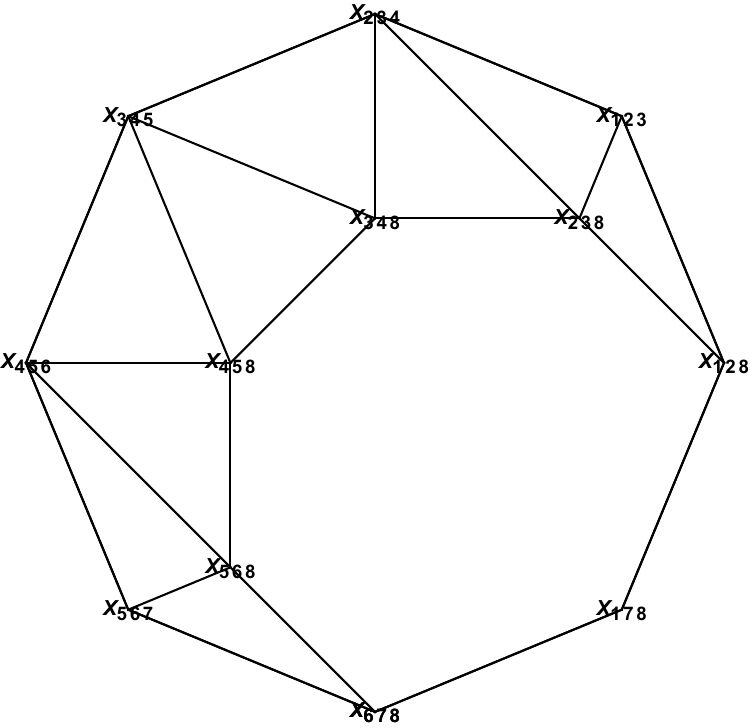}
	\caption{Combinatorial interpretation of a compatible collection of inverse propagators involved in the residual embedding of $m^{(2)}_{7}$ amplitude inside $m^{(3)}_{8}$. Deforming this with the mixed kinematic shift $\sigma$ in Equation \eqref{eqn: Xijk mixed deformation} and taking one more residue where $X_{168}=0$ or $X_{278}=0$ gives a residual embedding of the $6$-point NLSM amplitude into a mixed $8$-point CEGM amplitude.}
	\label{fig:residualembeddingnlsm6in8}
\end{figure}
For example, with $\sigma \in \mathcal{K}_{3,8}$ defined by 
$$X_{ijk} \mapsto X_{ijk} + \delta\left(\vert ijk \cap 1357\vert -1\right),$$
we would like to compute the leading order in the $\delta\rightarrow\infty$ limit and take the residue where $X_{238} = X_{348} = X_{458} = X_{568}=0$; instead, since the operations commute we first take the residue and then deform, and find 
\begin{eqnarray*}
	\text{Res}_{\bullet=0}(\mathcal{A}^{(3),\sigma}_8) &= &  2- \frac{X_{138}+X_{248}}{X_{148}}-\frac{X_{248}+X_{358}}{X_{258}}-\frac{X_{358}+X_{468}}{X_{368}}-\frac{X_{468}+X_{578}}{X_{478}}\\
	& - & \frac{X_{138}+X_{158}+X_{248}+X_{268}+X_{358}+X_{468}}{X_{168}}\\
	& - & \frac{X_{248}+X_{268}+X_{358}+X_{378}+X_{468}+X_{578}}{X_{278}}\\
	& + & \frac{\left(X_{158}+X_{268}\right) \left(X_{248}+X_{358}\right)}{X_{168} X_{258}}+\frac{\left(X_{268}+X_{578}\right) \left(X_{248}+X_{358}\right)}{X_{258} X_{278}}\\
	& + & \frac{\left(X_{138}+X_{268}\right) \left(X_{358}+X_{468}\right)}{X_{168} X_{368}}+\frac{\left(X_{268}+X_{378}\right) \left(X_{358}+X_{468}\right)}{X_{278} X_{368}}\\
	& + & \frac{\left(X_{138}+X_{248}\right) \left(X_{468}+X_{578}\right)}{X_{148} X_{478}} + \frac{\left(X_{248}+X_{378}\right) \left(X_{468}+X_{578}\right)}{X_{278} X_{478}}\\
	& + & \frac{\left(X_{138}+X_{248}\right) \left(X_{158}+X_{468}\right)}{X_{148} X_{168}}. 
\end{eqnarray*}
This is identified via $X_{ij} = X_{ij8}$ with the mixed 7-point amplitude induced by 
$$X_{ij} \mapsto X_{ij} + \delta(\vert ij\cap 1357\vert  -1).$$
Moreover, the residue where $X_{278}=0$ is identified with the 6-point NLSM amplitude,
\begin{eqnarray*}
	\text{Res}_{X_{168}=0}(\mathcal{A}^{(3),\sigma})_6 & = & -X_{138}-X_{158}-X_{248}-X_{268}-X_{358}-X_{468}\\
	& + & \frac{\left(X_{158}+X_{268}\right) \left(X_{248}+X_{358}\right)}{X_{258}}+\frac{\left(X_{138}+X_{248}\right) \left(X_{158}+X_{468}\right)}{X_{148}}\\
	& + & \frac{\left(X_{138}+X_{268}\right) \left(X_{358}+X_{468}\right)}{X_{368}}.
\end{eqnarray*}

For comparison,
\begin{eqnarray*}
	\mathcal{A}^{(2),NLSM}_6 & = & X_{13}+X_{15}+X_{24}+X_{26}+X_{35}+X_{46}-\frac{\left(X_{15}+X_{26}\right) \left(X_{24}+X_{35}\right)}{X_{25}}\\
	& - & \frac{\left(X_{13}+X_{24}\right) \left(X_{15}+X_{46}\right)}{X_{14}}-\frac{\left(X_{13}+X_{26}\right) \left(X_{35}+X_{46}\right)}{X_{36}}.
\end{eqnarray*}

In the next section we present a brief glimpse into factorization for $k=4$ generalized NLSM amplitudes.

\subsection{Factorizing an $8$-point NLSM $4$-Amplitude}\label{sec: factorization m48}

We consider the $\delta$-deformation 
\begin{eqnarray}\label{eqn: deformation 48}
	X_{ijk\ell} & \mapsto & X_{ijk\ell} + \delta\left(\vert 1357\cap ijk\ell\vert-2\right),
\end{eqnarray}
which determines the pure kinematic deformation $\sigma \in \mathcal{K}_{4,8}$ given by  
$$\sigma = -e^{1234}+e^{1238}-e^{1278}+e^{1678}+e^{2345}-e^{3456}+e^{4567}-e^{5678}.$$
Note that all $8$ terms appear.
We will argue that the residue of $m^{(4)}_8$ and of its generalized NLSM deformation in Equation \eqref{eqn: deformation 48}, where 
$$X_{1456},X_{1256},X_{1236}=0,$$
are respectively $m^{(4)}_6 \cdot m^{(4)}_6$ and $\mathcal{A}^{NLSM}_6\cdot \mathcal{A}^{NLSM}_6$.  Recall that $m^{(k)}_{k+2} \equiv m^{(2)}_{k+2}$.

Let us recall briefly the characterization of compatibility: two $k$-propagators $X_{I},X_J$ are \textit{compatible} provided that their sum, thought of as an element\footnote{Strictly speaking, modulo the lineality subspace.} of $\mathbb{R}^{\binom{n}{k}}$, satisfies the 3-term positive tropical Pl\"{u}cker relations, see Equation \eqref{eq: positive tropical Plucker relations}.  With this, we can proceed with our construction of the residue.

The CEGM amplitude $m^{(4)}_8$ is a known rational function in exactly 360 inverse $4$-propagators \cite{He:2020ray}; among these, exactly 18 are compatible with all three $X$-variables  $X_{1456},X_{1256},X_{1236}$, and they are split up into two groups of 9.  These two groups are mutually compatible and each forms a copy of $m^{(4)}_6$.

The first copy involves $X$-variables
$$X_{1235},X_{1245},X_{1246},X_{1345},X_{1346},X_{1356},X_{2346},X_{2356},X_{2456},$$
which are exactly those for $m^{(4)}_6$, while the second requires a little bit more work; its poles are encoded the 6 $X$-variables
$$X_{1237},X_{1267},X_{1268},X_{1567},X_{1568},X_{4568},$$
together with 
$$	X_{1236}-X_{1368}+X_{1378}+X_{1568},\ X_{1236}-X_{2368}+X_{2378}+X_{4568},\ X_{1256}-X_{2568}+X_{2678}+X_{4568}.$$
We choose an identification between these and the inverse propagators for $m^{(2)}_6$; the identification with $X_{(\mathbf{S},\mathbf{r})}$ is unique.  We have
$$
\begin{array}{ccc}
	X_{13} & X_{1268} & X_{(1278_3 3456_1)}\\
	X_{14} & X_{1568} & X_{(178_2 23456_2)} \\
	X_{15} & X_{4568} & X_{(123456_3 78_1)}\\
	X_{35} & X_{1256}-X_{2568}+X_{2678}+X_{4568} & X_{(12_1 78_1 3456_2)} \\
	X_{36} & X_{1267} & X_{(128_2 34567_2)}\\
	X_{46} & X_{1567} &  X_{(18_1 234567_3)}\\
	X_{24} & X_{1236}-X_{1368}+X_{1378}+X_{1568} & X_{(178_2 456_1 23_1)}\\
	X_{25} & X_{1236}-X_{2368}+X_{2378}+X_{4568}  & X_{(123_2 78_1 456_1)} \\
	X_{26} & X_{1237} & X_{(1238_3 4567_1)}\\
\end{array}
$$
The identification between the left and middle columns is based on the choice of a graph isomorphism of the compatibility complexes of the $X_{ij}$ and the 9 inverse $4$-propagators above.  The identification between columns 2 and 3 is an identity modulo momentum conservation based on Appendix \ref{sec: two formulas for k propagators}.

Most importantly, we notice that the deformations 
$$X_{ij} \mapsto X_{ij} + \delta(\vert ij \cap 246\vert -1)$$
and 
$$X_{ijk\ell} \mapsto X_{ijk\ell} + \delta(\vert ijkl \cap 1357\vert -2)$$
act identically on the left and middle columns, respectively, and we have found an identification between our residue of $\mathcal{A}^{(4),\sigma}_8$ and the 6-point NLSM amplitude.

\section{Hard and Soft Limits}\label{sec: hard and soft limits}
In this section, we show that the generalized NLSM amplitude from Equation \eqref{eq: generic GNLSM Intro A}, given by 
\begin{eqnarray}\label{eq: generic GNLSM Intro}
	\mathcal{A}^{(3),\sigma}_{6}& = & X_{125} X_{134}+X_{146} X_{134}+X_{124} X_{145}+X_{136} X_{145}+X_{136} X_{235}+X_{145} X_{235}\nonumber\\
	& +& X_{146} X_{236}+X_{245} X_{236}+X_{125} X_{245}+X_{146} X_{245}+X_{124} X_{256}+X_{235} X_{256}\nonumber\\
	& + & X_{124} X_{346}+X_{136} X_{346}+X_{256} X_{346}+X_{125} X_{356}+X_{134} X_{356}+X_{236} X_{356}\nonumber\\
	& -& \frac{\left(X_{124}+X_{236}\right) \left(X_{146}+X_{256}\right) \left(X_{245}+X_{346}\right)}{X_{246}}\nonumber\\
	& - & \frac{\left(X_{125}+X_{136}\right) \left(X_{134}+X_{235}\right) \left(X_{145}+X_{356}\right)}{X_{135}}\\
	& - & \frac{\left(X_{134}+X_{256}\right) \left(X_{125}+X_{346}\right) \left(X_{124}+X_{356}\right)}{X'_{246}}\nonumber\\
	& - & \frac{\left(X_{146}+X_{235}\right) \left(X_{145}+X_{236}\right) \left(X_{136}+X_{245}\right)}{X'_{135}},\nonumber
\end{eqnarray}
vanishes identically in all hard and soft limits; we now describe the computation using certain hard and soft boundary operators, respectively $\partial_{\ell_1}$ and $\partial_{\ell_0}$.  The general structure mirrors the face lattice of the hypersimplex $\Delta_{k,n}$, whose $2n$ facets are $n$ copies of $\Delta_{k-1,n-1}$ and $\Delta_{k,n-1}$, where $x_j=1$ and $x_j=0$, respectively.  We point out that soft limits for CEGM amplitudes have been studied in \cite{GarciaSepulveda:2019jxn}.

For each $\ell=1,\ldots, n$, introduce the following ``soft'' degeneration:
\begin{eqnarray}
	\mathfrak{s}^{(\ell_0)}_J(\lambda) & = & \begin{cases}
		\mathfrak{s}_J,\ \ell\not\in J,\\
		\lambda \mathfrak{s}_J,\  \ell \in J.
	\end{cases}
\end{eqnarray}
Similarly, for each $\ell=1,\ldots, n$, introduce the following ``hard'' degeneration:
\begin{eqnarray}
	\mathfrak{s}_J^{(\ell_1)}(\lambda) & = & \begin{cases}
		\mathfrak{s}_J,\ \ell\in J,\\
		\lambda \mathfrak{s}_J,\  \ell \not\in J\\
	\end{cases}
\end{eqnarray}
In the $X$-basis, the hard degeneration has a particularly simple expression
\begin{eqnarray}
	X_{J}^{(\ell_1)}(\lambda) & = & \lambda X_J + (1-\lambda)X^{(\ell_1)}_{J'}
\end{eqnarray}
where $X^{(\ell_1)}_{J'}$ is an $X$-variable on the kinematic subspace
$$\mathcal{K}_{k-1,n-1} \simeq \left\{(\mathfrak{s})\in \mathcal{K}_{k,n}: s_{J}=0\text{ whenever } \ell \not\in J\right\},$$
and where $J'$ is obtained from $J$ by deleting the cyclic successor in $J$ to $\ell$, see \cite{Early:2022zny}.  When clear from context, we write simply $X^{(\ell)}_{J'}$ for the hard degeneration in the particle $\ell$.  If $J'$ is a cyclic interval in $\lbrack n\rbrack\setminus \{\ell\}$, it follows that $X^{(\ell_1)}_{J'}\equiv 0$.  The \textit{hard limit} (of $m^{(k)}_n$) in the particle $\ell$ is the leading-order coefficient in its Laurent expansion in $\lambda$ as $\lambda \rightarrow0$.  We define $\partial_{\ell_1}(X_J) = X^{\ell_1}_J(\lambda=0)$.

For example, for $(k,n) = (3,6)$, the hard limit in particle $1$ acts as follows:
\begin{eqnarray*}
	\partial_{1} (\mathfrak{s}_{123} = X_{236}) &=& X^{(1)}_{36},\\
	\partial_{1} (\mathfrak{s}_{234} = X_{134}) &=& 0,\\
	\partial_{1} (\mathfrak{s}_{156} = X_{146}) &=& X^{(1)}_{46},\\
	\partial_{1} (\mathfrak{s}_{123} + \mathfrak{s}_{124} + \mathfrak{s}_{134} + \mathfrak{s}_{234} = X_{346}) &=& X^{(1)}_{46},\\
	\partial_{1} (X_{246}) &=& X^{(1)}_{46}.
\end{eqnarray*}
where $X_{246}$ is one of the first surprises encountered in the study of CEGM amplitudes.  It has a somewhat un-enlightening expression indirectly in terms of the $\mathfrak{s}_{ijk}$'s, 
$$X_{246} = s_{156}+s_{256}+s_{345}+s_{346}+s_{356}+s_{456},$$
but the formula in Appendix \ref{sec: two formulas for k propagators} provides more insight into its structure:
$$X_{246} = X_{(12_1 34_1 56_1)}.$$
Here, we have an over all 3-fold cyclic symmetry and a flip symmetry in each block, for example
$$X_{(12_1 34_1 56_1)} = X_{(12_1 34_1 65_1)} = X_{(56_1 12_1 34_1)} = \cdots,$$
reflecting that the residue $X_{(12_1 34_1 56_1)}=0$ of $m^{(3)}_6$ is a triple product of three 4-point amplitudes.

Observe that all of the following inverse $3$-propagators are identically zero in the hard limit $\partial_{1_1}$:
\begin{eqnarray}\label{eq: hard limit 1}
	X_{134},X_{145},X_{245},X_{256},X_{356}\mapsto 0,
\end{eqnarray}
and from Equation \eqref{eq: generic GNLSM Intro} we see that $\mathcal{A}^{(3),\sigma}_{6}$ remains finite when Equation \eqref{eq: hard limit 1} is imposed.

Let us now compute the hard limit of $\mathcal{A}^{(3),\sigma}_6$ in two steps for pedagogical purposes.  We first impose Equation \eqref{eq: hard limit 1} and then apply $\partial_{1_1}$ explicitly.  We first observe that $\mathcal{A}^{(3),\sigma}_6$ remains finite when $\lambda=0$, so we can simply apply $\partial_{1_1}$ to the whole expression and simplify, getting
\begin{eqnarray*}
	& & \partial_{1_1} \left(\mathcal{A}^{(3),\sigma}_{6}\right)\\
	& = & \partial_{1_1} \left(X_{136} X_{235}+X_{146} X_{236}+X_{124} X_{346}+X_{136} X_{346} - \frac{X_{136} \left(X_{146}+X_{235}\right) X_{236}}{-X_{135}+X_{136}+X_{235}}\right)\\
	& - & \partial_{1_1} \left(\frac{X_{146} \left(X_{124}+X_{236}\right) X_{346}}{X_{246}}\right)\\
	& = & X^{(1)}_{35} X^{(1)}_{36}+X^{(1)}_{46} X^{(1)}_{36}+ X^{(1)}_{46} X^{(1)}_{36} + X^{(1)}_{24} X^{(1)}_{46}- \left(X^{(1)}_{35}+X^{(1)}_{46}\right) X^{(1)}_{36}-\left(X^{(1)}_{24}+X^{(1)}_{36}\right) X^{(1)}_{46}\\
	& = & 0,
\end{eqnarray*}
having made use of identities such as 
$$\partial_{1_1}(X_{146})=\partial_{1_1}(X_{246}) = \partial_{1_1}(X_{346}) = X^{(1)}_{46},\ \ \partial_{1_1}(X_{135})=\partial_{1_1}(X_{235}) = X^{(1)}_{35},$$
and
$$\partial_{1_1}(X_{135})=\partial_{1_1}(X_{235}) = X^{(1)}_{35}.$$
Finally, since $m^{(3)}_6$ is sent to itself under the duality $m^{(k)}_{n} \mapsto m^{(n-k)}_n$ which also interchanges hard and soft boundaries $\partial_{\ell_0} \leftrightarrow \partial_{\ell_1}$, we conclude that the soft limit vanishes as well.

While we do not have substantial experimental evidence beyond $(k,n) = (3,6)$, it seems more likely than not that generic, pure kinematic deformations will have an analog of the Adler zero.

\begin{conjecture}
	For a generic, pure generalized NLSM $3$-amplitude, all single soft limits vanish identically.
\end{conjecture}

\section{Discussions and Directions For Future Work}\label{sec: discussions}

In this paper, we have introduced a systematic deformation theory for CEGM amplitudes, developed zero-preserving deformations that produce generalized NLSM amplitudes.  We proved that the NLSM emerges as a certain residues of a mixed $(3,n+2)$ CEGM NLSM amplitude.

Our computations of soft and hard limits further reveal that pure deformations have a nontrivial structure and we found evidence of an analog of the Adler zero, wherein the amplitude vanishes when one particle is soft.  We find that the space of zero-preserving deformations of CEGM amplitudes have a much richer structure than what is present for $k=2$.

For $\text{Tr} (\phi^3)$, the ``color'' $U(N)$ degrees of freedom play essentially only an organizational role in the amplitudes, telling us to work with ``ordered'' amplitudes. For the NLSM, the unitary group $U(N)$ plays a much deeper role: physics tells us that the pions are the Goldstone bosons of the $U(N) \times U(N) \to U(N)$ symmetry breaking pattern, and this is reflected in the properties of the soft and multi-soft limits.  Of course, for gluons, the $U(N)$ symmetry is even deeper still -- it determines the Yang-Mills gauge symmetry crucially underpinning the dynamics, and is reflected in the properties of the amplitudes in a myriad of ways, largely determining the amplitudes.  It is a fascinating question to explore how these incarnations of the $U(N)$ group structures are generalized in CEGM.

The CEGM construction has provided a powerful new apparatus for studying scattering amplitudes, and has inspired many new developments and connections between scattering amplitudes and many other research areas.  There has been a substantial amount of work; let us try to provide a representative list of topics which provide a roadmap for how the subject has grown and how it may continue to evolve.  These research topics include: configuration spaces \cite{Early:2019xbh}, combinatorial approaches to leading singularities associated to non-planar on-shell diagrams \cite{Cachazo:2018wvl}, cluster algebras \cite{Drummond:2019qjk,Ren:2021ztg,Henke:2021ity}, algebraic statistics \cite{Sturmfels:2020mpv,Agostini:2021rze}, amplituhedra \cite{Lukowski:2020dpn}, quantum affine algebras \cite{Early:2023tjj}, and new phenomena and research areas such as minimal kinematics \cite{Cachazo:2020uup,Early:2018zuw,Early:2024nvf,Arkani-Hamed:2023swr,Arkani-Hamed:2024fyd}, smooth splits and split kinematics \cite{Cachazo:2021wsz,Zhang:2024iun,GimenezUmbert:2025ech}, positive del Pezzo geometry \cite{Early:2023cly}, applications to $\mathcal{N}=4$ SYM \cite{Arkani-Hamed:2019rds}, Grassmannian string integrals \cite{He:2020ray,Arkani-Hamed:2019mrd} and spinor-helicity varieties \cite{Maazouz:2024qmm}, and compatibility complexes for Minkowski sums of alcoved polytopes \cite{Early:2025hqs}.  

We are steadily digging deeper and deeper into these structures, revealing a generalization of QFT that we are seeing emerging from CEGM amplitudes, suggesting an actual connection to real-world physics.  Many fascinating questions remain to be explored.  In future work we will extend our study of the CEGM NLSM to $k\ge 4$, revealing an interconnected web of theories.  We expect significant progress on these issues in the next year.

\section*{Acknowledgements}
We thank Nima Arkani-Hamed for discussions and valuable comments on the draft.  We also thank Carolina Figueiredo, Alex Fink, Hadleigh Frost, Alfredo Guevara, Thomas Lam, Sebastian Mizera, Hugh Thomas and Bruno Umbert for discussions and helpful comments, and especially Freddy Cachazo for enlightening discussions during the early stages of the work.  We are grateful to Perimeter Institute for Theoretical Physics for excellent working conditions while this project was initiated.  The author was funded by the European Union (ERC, UNIVERSE PLUS, 101118787). Views and opinions expressed are however those of the author(s) only and do not necessarily reflect those of the European Union or the European Research Council Executive Agency. Neither the European Union nor the granting authority can be held responsible for them.

\appendix

\section{Parametrization of $X^+(k,n)$}\label{sec:positive parametrization}
In this Appendix we construct a parameterization of the positive configuration space $X^+(k,n)$.  Substituting $x_{i,j} \mapsto \prod_{t=0}^j x_{i,t-1}$ and setting $x_{i,0}=1$ relates our parameterization to the standard web parameterization.  

We first define a $(k-1)\times (n-k)$ polynomial-valued matrix $M_{k,n}$ with entries $m_{i,j}(x)$, with $(i,j) \in \lbrack 1,k-1\rbrack \times \lbrack 1,n-k\rbrack$, defined by 
\begin{eqnarray*}
	m_{i,j}(\{x_{a,b}: (a,b) \in \lbrack i,k-1\rbrack \times \lbrack 1,j\rbrack\}) & = & (-1)^{k-i}\sum_{1\le b_i\le b_{i+1}\le \cdots \le b_{k-1}\le j}\left(x_{i,b_{i}}x_{i+1,b_{i+1}}\cdots x_{k-1,b_{k-1}}\right).
\end{eqnarray*}

For instance, fixing $(k,n) = (4,n)$ then 
\begin{eqnarray*}
	m_{1,2} & = & -(x_{1,1} x_{2,1} x_{3,1}+x_{1,1} x_{2,1} x_{3,2}+x_{1,1} x_{2,2} x_{3,2}+x_{1,2} x_{2,2} x_{3,2})\\
	m_{2,3} & = & x_{2,1} x_{3,1}+x_{2,1} x_{3,2}+x_{2,2} x_{3,2}+x_{2,1} x_{3,3}+x_{2,2} x_{3,3}+x_{2,3} x_{3,3}\\
	m_{3,4} & = & -(x_{3,1}+x_{3,2}+x_{3,3}+x_{3,4}).
\end{eqnarray*}
For sake of brevity, in what follows we denote $x_{i,S} = \sum_{j\in S}x_{i,j}$ for any $S\subseteq \lbrack 1,n-k\rbrack$.

We construct a $k\times n$ matrix $M$ with $M_{k,n}$ as its upper right $(k-1)\times (n-k)$ block:
\begin{eqnarray}\label{eq: pos parametrization}
	M & = & \begin{bmatrix}
		1 &  &  & 0 & m_{1,1}& \cdots  & m_{1,n-k} \\
		& \ddots &  &  &\vdots & \ddots & \vdots \\
		&  & 1 &  & m_{k-1,1} & & m_{k-1,n-k} \\
		0&  &  & 1 & 1 & \cdots & 1 \\
	\end{bmatrix}.
\end{eqnarray}
For instance, for $k=3$ we have
\begin{eqnarray}
	M_{3,6} & = & \begin{bmatrix}
		m_{1,1} & m_{1,2} & m_{1,3}\\
		m_{2,1} & m_{2,2} & m_{2,3}
	\end{bmatrix}\\	
	& = & \begin{bmatrix}
		x_{1,1}x_{2,1} & x_{1,1}x_{2,12}+x_{1,2}x_{2,2}  & x_{1,1}x_{2,123} + x_{1,2}x_{2,23} + x_{1,3}x_{2,3} \\
		-x_{2,1} & -x_{2,12} & -x_{2,123}\nonumber
	\end{bmatrix}
\end{eqnarray}
and for the embedding we have
$$M = \begin{bmatrix}
	1 & 0 & 0 & x_{1,1}x_{2,1} & x_{1,1}x_{2,12}+x_{1,2}x_{2,2} & x_{1,1}x_{2,123} + x_{1,2}x_{2,23} + x_{1,3}x_{2,3} \\
	0& 1 & 0 & -x_{2,1} & -x_{2,12} & -x_{2,123}\\
	0 & 0 & 1 & 1 & 1 & 1 
	\nonumber
\end{bmatrix}.$$

\bibliographystyle{jhep}
\bibliography{references}
\end{document}